\documentclass{llncs}
\usepackage{fullpage}
\usepackage{amsmath}
\usepackage{amsfonts}
\usepackage{amssymb}
\usepackage{graphicx}
\usepackage{color}
\usepackage{textcomp} 
\usepackage{nicefrac}
\usepackage{wrapfig}
\usepackage{multirow}
\usepackage{bigstrut}
\usepackage{longtable}

\newcommand{\remove}[1] {}

\newcommand{\eps}{\varepsilon}
\renewcommand{\epsilon}{\varepsilon}
\newcommand{\etal}{{\rm et~al.}}

\newlength{\pgmtab}
\newcommand{\BCP}{{\sc BCSP\-(price)}}
\newcommand{\BCB}{{\sc BCSP\-(bid)}}
\newcommand{\BCBO}{{\sc BC\-SP\-(best\- offer)}}
\newcommand{\BOSP}{{\sc BOSP}}
\newcommand{\CTR}{{\sc CTR}}
\newcommand{\GSP}{{\sc GSP}}

\newcommand\blfootnote[1]{%
  \begingroup
  \renewcommand\thefootnote{}\footnote{#1}%
  \addtocounter{footnote}{-1}%
  \endgroup
}

\begin{document}

\title{On the Stability of Generalized Second Price Auctions with Budgets}

\author{Josep D{\'i}az\inst{1} \and Ioannis Giotis\inst{1,2} \and Lefteris Kirousis\inst{3} \and Evangelos Markakis\inst{2} \and Maria Serna\inst{1}}

\institute{Departament de Llenguatges i Sistemes Informatics\\ Universitat Politecnica de Catalunya, Barcelona
\and 
Department of Informatics \\ Athens University of Economics and Business, Greece
\and 
Department of Mathematics,  National \& Kapodistrian University of Athens, Greece 
\\ and Computer Technology Institute \& Press  ``Diophantus''}

\maketitle
\begin{abstract} 
The Generalized Second Price (GSP) auction used typically to mo\-del sponsored search auctions does not include the notion of budget constraints, which is present in practice. Motivated by this, we introduce the different variants of GSP auctions that take budgets into account in natural ways. We examine their stability by focusing on the existence of Nash equilibria and envy-free assignments. We highlight the differences between these mechanisms and find that only some of them exhibit both notions of stability. This shows the importance of carefully picking the right mechanism to ensure stable outcomes in the presence of budgets.\protect\blfootnote{\scriptsize Contact: \email{\{diaz,igiotis,mjserna\}@lsi.upc.edu,lkirousis@math.uoa.gr,markakis@gmail.com}}\blfootnote{\scriptsize Josep D{\'i}az, Maria J. Serna and Ioannis Giotis supported by the CICYT project TIN-2007-66523 (FORMALISM). This research has also been co-financed by the European Union (European Social Fund – ESF) and Greek national funds through the Operational Program \lq\lq{}Education and Lifelong Learning\rq\rq{} of the National Strategic Reference Framework (NSRF) - Research Funding Program: Thales. Investing in knowledge society through the European Social Fund.}
\end{abstract}




\pagestyle{plain}

\section{Introduction}

Advertising on Internet search engines has evolved into a phenomenal driving force both for the search engines and the advertising businesses. It is a modern and rapidly growing method that is now being implemented in various other popular sites beyond search engines, such as blogs, and social networking sites. Although some rightful concern has been raised regarding privacy issues and distinguishability from non-sponsored results, there are clear advantages to the advertisers who can efficiently reach their target audiences and observe the results of their ad campaign within days or even hours. At the same time, online ads account for a large share of the profits for search engines and other participating web-sites. Even the web-user experience can be enhanced, by the delivery of additional information relevant to their queries.

In a typical instance, a user queries a search engine for a particular keyword of commercial interest, and the search engine determines the ads to be displayed by means of an auction. The prevailing system uses a pay-per-click policy, i.e., it only charges an advertiser when the user clicks on the corresponding link and is diverted to the advertiser' s web-site. 

The mechanism used can be viewed as an auction for multiple homogenous indivisible items, the advertisement areas available, with single-demand buyers, since it is not desirable for the same advertisement to appear more than once. Such auctions can find applications in a wide variety of scenarios besides Internet advertising where buyers have a valuation per unit of a particular good but the setting is restricted to selling only single fixed-sized bundles. For example, consider selling different fixed-sized shipments of a food product when the seller cannot send more than one shipment to the same destination or frequency spectrum auctions of different fixed-sized bandwidths where regulations do not permit buying more than one continuous bandwidth. Noting that the mechanisms can be applied in much different scenarios, we will work in the context of Internet search advertising both because of its wide-spread application today and also because of the significant focus it receives in related literature.

In the early history of sponsored search auctions, the allocation of slots to advertisers was determined by a \emph{first-price auction}, as in the systems originally used by Overture. Later on, Google was the first to switch to a \emph{second-price auction}, an approach which demonstrated superior characteristics and was quickly adopted by the rest of the major search engines. The main idea is that the advertisers declare how much they are willing to pay for a click to their ad but they are charged instead a lesser amount equal to the next lower competing bid. Apart from its elegant simplicity, this scheme has been quite successful in terms of its generated revenue as well. In the literature, this system is commonly known as the \emph{Generalized Second-Price} (\GSP) auction and, by now, a large volume of work has emerged on the study of \GSP ~auctions and related mechanisms, see e.g., the surveys \cite{LPSV07} and \cite{MMNST12}.

However, 
not all is positive about this mechanism. It was quickly pointed out that the mechanism was not \emph{truthful}, i.e. it is not in an advertiser' s best interest to declare his true intention as to how much he is willing to pay. Immediately, this raises the troublesome question as to how should an advertiser behave, a question which proved to be significantly complex to resolve. Erratic behavior has been found to be damaging for both the advertisers who cannot clearly evaluate and plan their campaign' s performance and the search engine which sees its own profits declining~\cite{EO07}. Stability in such systems is thus extremely desired, where the advertisers are happy with their choices and the whole system remains trackable to monitor and evaluate.
an aspect that has been often ignored, especially in the early literature, is the presence of a budget constraint, requested from the advertisers to limit their exposure and expenditure. We believe this is a key parameter, essential in accurately understanding and evaluating the systems used in practice.

In the last years, there have been more attempts to take budgets into account, as we explain in the Related Work Section. This has led to some recently proposed mechanisms, which however tend to be complex, they lack the simplicity of \GSP\ auctions and are unlikely to be implemented in practice. Hence, a question that still remains unanswered is whether we can have simple second price mechanisms, that account for budgets and still possess desirable properties with respect to the stability of their outcomes. This is a question of general interest, beyond sponsored search auctions, since budget constraints can be present in other mechanism design problems as well. 

\medskip
\noindent{\bf Our contribution.}
Our main conceptual contribution is a study of Generalized Second-Price auctions under the presence of budgets. First we showcase that ignoring these constraints might lead to unstable outcomes. We then introduce three simple and natural extensions of the GSP mechanism, that take budgets into account. As it is not straight-forward to define a single natural mechanism, we define these variants motivated by the key desirable properties of second-price auctions. For all mechanisms, we investigate the existence of Nash equilibria and envy-free assignments, which are the main notions of stability that have been considered in the literature. For the first mechanism, we show that a Nash equilibrium might not always exist, yet an envy-free assignment is always achievable. For the other two mechanisms, we show that they always possess envy-free Nash equilibria, in fact we show that any envy-free assignment can be realized as an equilibrium of the mechanisms under consideration. In our model, we consider the budget as part of a bidder's strategy, i.e., we have {\em private} budgets. An interesting and surprising outcome of our study is that in the case of public budgets Nash equilibria do not always exist, despite the existence of envy-free assignments. In contrast to mechanism design problems, where having public budgets usually eases the design of a truthful algorithm, here we realized that having public budgets may eliminate the existence of stable profiles. 
Of separate interest might be an algorithmic process than can construct an envy-free assignment in our model.
Overall, we believe our work can serve as a starting point for studying further the properties of GSP auctions that take budgets into account.

\subsection{Related work}

Varian~\cite{varian2007position} and Edelman \etal~\cite{Edelman} have been the two seminal works on equilibrium analysis of \GSP\ auctions without budgets. 
They established the existence of a Nash equilibrium which also satisfies other desirable stability properties such as being welfare-maximizing and envy-free. A further analysis of envy-free Nash equilibria, by taking into account the quality factor of the advertisers was also provided in \cite{lahaie2007revenue}.

The notion of budget constraints has been introduced in various models and objectives, such as, among others, in~\cite{borgs2005}, \cite{Chakrabarty2007budget}, \cite{FMPS07}, \cite{Charles13} and \cite{Fiat11}. Recent work on truthful mechanism design, mainly inspired by the clinching auction of Ausubel~\cite{ausubel2004efficient}, has led to the introduction of truthful Pareto-optimal mechanisms in the presence of budgets, see e.g. \cite{Dobzinski08,GML12,ColiHLS12}. 
Ashlagi \etal introduced the model we'll be using in~\cite{ABHLT13}. 
However, all these mechanisms employ techniques that are very different from the {\GSP} scheme in order to achieve truthfulness. As a result, they lack the simplicity of second price auctions at the expense of achieving better properties.

The work of Arnon and Mansour~\cite{arnon2011repeated} is conceptually closer to our approach. They studied second-price auctions with budget constraints but their model simplifies the items for sale to clicks, as opposed to the slots, allowing a player to potentially receive more or less clicks than a single slot could offer. This deviates from the one player per slot paradigm used in practice.

Finally, a different direction that has been pursued recently is the performance of mechanisms in terms of the generated social welfare.
Price of Anarchy analysis for auctions was initiated in \cite{CKS08}, see also \cite{CKK+12}, for sponsored search auctions without budgets. For certain settings with budget constraints, some results have been recently obtained in \cite{ST13} (which however do not have any implications for our proposed mechanisms). Our work does not focus on Price of Anarchy, which we leave for future research, but on existence of stability concepts.
\section{Preliminaries}
Our model is the same as in Ashlagi \etal~\cite{ABHLT13}, a natural extension to budget limited players of the model introduced by Varian~\cite{varian2007position} which is widely adopted in related literature.

We assume we have $k$ slots, each with a fixed, distinct\footnote{This assumption is derived from the distinct space these slots occupy on a web-page.} and publicly known click-through rate (\CTR), $\theta_{j}$ for slot $j$, representing the number of clicks received in a fixed time period (typically a day), independently of the advertisement displayed. Let us order the slots such that $\theta_1> \theta_2 > \ldots > \theta_k$. 
Even though the click-through rates are probabilistic in nature, we will make the typical assumption that they are \emph{deterministically} realized for simplification purposes; $\theta_i$ will really correspond to the expected click-through rate of slot $i$. 
Contrary to the numerical ordering, we will typically use the terminology ``higher'' and ``lower'' slots referring to slots of higher and lower \CTR. Finally, for ease of illustration, we will ignore the bidder-dependent quality factor that is usually incorporated in calculating click-through rates in the separable model.

We have $n\geq k$ players (advertisers). Each player $i$ has a private \emph{valuation} $v_i$ representing the perceived value per click. Each player also has a \emph{budget constraint} $B_i$, indicating the total amount he is willing to spend in a fixed time period, not on a per click basis. We will also assume that these budget values are pairwise distinct. This assumption has been necessary in other works as well \cite{ABHLT13}, and affects many properties in related mechanisms~\cite{Dobzinski08,ABHLT13}. In fact, as we will exhibit later on, envy-free assignments, which is one of the stability concepts we are interested in, are not guaranteed to exist when budgets are not distinct. To see this assumption is necessary,
consider the example in Figure~\ref{fig:notEF} and assume we have two slots with $\theta_1=1$ and $\theta_2=0.5$. 
The main idea behind this example is that when the slots become affordable, prices have already come down too low, and therefore envy will arise. Interestingly enough, notice that if we differentiate the budgets just infinitesimally, the counterexample seizes to hold.

\begin{figure}[t]\center
\begin{tabular}{|c|c|c|}
\hline  & Player 1 & Player 2 \\
\hline Value & 8  & 6  \\
\hline Budget & 2 & 2 \\
\hline
\end{tabular}
\caption{Example of an instance with equal budgets where no envy-free assignment exists. In this example we have two slots with $\theta_1 = 1$, $\theta_2 = 0.5$.}
\label{fig:notEF}
\end{figure}
Hence, similarly to previous works, we also choose to adopt the distinctness of budgets.

Each player $i$ is interested in maximizing $\theta_{s(i)} (v_i - p(i))$, where $s(i)$ is the slot assigned to $i$ and $p(i)$ the accompanying price per click requested by the respective mechanism. At the same time, he must also satisfy the budget constraint, $\theta_{s(i)} p(i)\leq B_i$. If this condition holds, we say the player can afford slot $s(i)$. 
We wish to enforce strict budget constraints so we define the utility of the players whose budget constraints are violated to be minus infinity as is typically done in the literature; any other negative value would also serve our purpose (i.e., budget violations are less desirable than not getting a slot). More formally,
\[ u_i = \left\{ \begin{tabular}{ll} $0$, & if $i$ was not awarded a slot,\\$\theta_{s(i)}(v_i - p(i)),$& if $\theta_{s(i)} p(i)\leq B_i$,\\$-\infty,$& otherwise.\end{tabular}\right. \]

\subsection{Second-Price Auctions under Budgets}

The players submit \emph{value-bids} $b_i$, representing the maximum amount they are willing to pay per click. These bids do not necessarily form a truthful declaration of the players' values to the mechanism. Similarly, the players also submit a \emph{budget-bid} $g_i$ to declare their budget. We will use the term bid to refer to the combination of these two types or to one particular type when clear from context.  In the case of ties, we assume there exists a fixed a priori defined ordering of the players based on which tie-breaking is resolved.

It is not trivial to introduce mechanisms that take budgets into account in a straightforward and natural way. To address this, we first ponder what constitutes a second-price mechanism by noting some key properties of generalized second-price auctions:
\begin{itemize}
\item[$\bullet$] The slot allocation should be performed by a simple and efficient process. 
\item[$\bullet$] The allocation should be in accordance with the bid ordering. If a player raises his bid he should be getting at least the slot he was getting before and should he lower his bid he should be getting at most the previous slot. 
\item[$\bullet$] Furthermore, if a bidder raises his bid, this should cause his total payment to potentially rise and respectively lowering his bid potentially lowers his payment. 
\item[$\bullet$] Finally, the price per click for each slot should be determined by either the next lower bid, the bid of the player awarded the next slot or the minimum bid required to obtain the slot. While these three concepts coincide in the regular {\GSP} mechanism, this is not the case when one introduces budgets. 
\end{itemize}

We first consider the {\GSP} mechanism without budgets in our context, only to highlight that ignoring budgets can lead to unstable outcomes. For notational consistency with the other mechanisms, we refer to this mechanism as Budget-Oblivious.
\begin{definition}[Budget-Oblivious]
The \emph{{budget-oblivious}} second-price auction, in short, \BOSP, orders the value-bids in decreasing order and then assigns the slots in that order, \emph{ignoring} the budget constraints. Naturally, the price for each slot is determined by the immediately lower value-bid.
\end{definition}
In this mechanism, the prices are in decreasing order however some budget constraints might be violated and that responsibility is passed on to the players. Note that strategic budget-bidding is not of interest in our mechanism since this type of bids is ignored by the mechanism.

Since \BOSP, as is shown in Section \ref{sec:BOSP}, does not have good stability properties, we turn our attention on mechanisms that respect the budget constraints by not assigning slots/prices to players that can afford them as declared by their budget-bids. The first interpretation of second-price pricing, charging the next lower bid, leads us to the following mechanism.


\begin{definition}[Budget-Conscious by Price] We define the \emph{budget-\-con\-scious by price} second-price auction, in short, \BCP, as the mechanism which first orders the value-bids in decreasing order and assigns a price per click for each player equal to the immediately lower value-bid in the bid ordering. Then, \BCP\ assigns the players to slots in order of decreasing value-bids, \emph{respecting} the budget constraints of each player as declared by their budget-bids, by assigning each player to the highest unassigned slot he can afford with his assigned price. If the player cannot afford any slot, he is left unassigned and he is not charged anything.
\end{definition}
Note that under this mechanism a player assigned to a slot might end up paying more per click than a player in a higher slot but we are guaranteed that all budget constraints of assigned players are satisfied. Also note that some slots might end up unassigned if no player can afford to occupy them. 

The mechanism above is a natural way to guarantee budget compliance but raises a fairness issue, as players might be declaring value-bids as the maximum amount they are willing to pay and getting a slot that they cannot afford to pay if they were to pay their own bid. As will be evident in the later sections, this can lead to players intentionally raising their bid to just below their competitor's bid. The following mechanism addresses this.
\begin{definition}[Budget-Conscious by Bid]
We define the \emph{{budget-conscious by bid}} second-price auction, in short, \BCB, similarly to \BCP\ except the mechanism now requires the players to be able to afford their slot if they were to pay a price per click equal to their own value-bid. The players are ordered in decreasing order of value-bids and the price of each player is set to the next lower bid. Then, the players are assigned from the highest bidder to the lower, one by one, to the highest available unassigned slot that they can afford should they were to pay their own bid. 
\end{definition}
The second way of interpreting second-price prices, charging prices equal to the value-bid of the player ending up occupying the next slot, first, has definitional issues since we cannot know if our player can afford a slot without knowing who gets the next one and secondly, might charge a player more than his bid. For these reasons, we do not investigate mechanisms of this type in this work.

Finally, we introduce a mechanism that considers what the players are willing and afford to pay for a slot by considering the minimum as implied from their value and budget bids. This mechanism essentially captures pricing by charging the minimum amount required to obtain a slot.
\begin{definition}[Best Offer Budget-Conscious]
We define the \emph{{best offer bud\-get-conscious}} second-price auction, in short, \BCBO, as the mechanism which intuitively awards each slot to the player that can offer the most ``money'' but charges them the next lower amount offered. More formally, each slot $s$, one by one from higher to lower, is awarded to the unassigned player with the largest $min \{b_i, g_i/\theta_s\}$  and he is charged a price per click equal to the second largest such value among unassigned players. We note that under this mechanism, the price charged for each slot is the minimum bid required to secure the slot. Alternatively, one can think of the slot rewarded to the player with the highest $min \{\theta_s b_i, g_i\}$, and paying the second highest such amount, representing the total offer of the player and the total price charged.
\end{definition}
It should be pointed out that in \BCBO, in turn of decreasing \CTR, is offered to the player who  is willing to pay the most, given that does not violate  his declared budget;  he is then charged what the second such player would pay (not counting players who already got a slot). Whereas in the previous two budget-conscious mechanisms, each player, in turn of its bid, chooses the best object he can afford and pays the bid  of the next player in line. The above two approaches are obviously equivalent in any  mechanism that does not refuse giving a slot to a player who cannot afford it. However, once we introduce into the mechanism the additional requirement of refusing to give objects to anybody who cannot afford it, then the above distinction becomes necessary. 

In all mechanisms, we assume that players not awarded a slot are not charged a payment and that the price of the lowest bidding player is zero, should he be awarded a slot. Finally, given a finite set of players, we note that the allocation of slots and pricing can be determined efficiently in all defined mechanisms.

\subsection{Stable assignments}
\label{subsec:stable}

It is easy to see that none of the mechanisms defined above are \emph{incentive-compatible}. There are cases where a player might receive a higher utility by ``lying'' about his value and getting a lower slot at a beneficial price, even if other players are truthfully bidding their values. Naturally, we turn our attention to notions of stability, a requirement to analyze significant properties of these auctions and in general a desired property for the advertisers as well. As usual, we will focus on the notion of \emph{Nash equilibrium}.
\begin{definition}[Nash Equilibrium]
A profile of bids, $\langle b_i,g_i \rangle$ for each player $i$, forms a Nash equilibrium if no player has an incentive to deviate to a different strategy $\langle b_i',g_i'\rangle$, for any $\langle b_i',g_i'\rangle$. 
\end{definition}
In related work~\cite{varian2007position,Edelman}, the notion of \emph{symmetric} or \emph{envy-free} equilibrium was defined. Under the generalized second-price auction without budget constraints, this class of envy-free equilibria is a subset of Nash equilibria. In the presence of budgets, this still holds for \BOSP, \BCB and \BCBO, but it does not hold for \BCP\ as shall be demonstrated later.
\begin{definition}[Envy-Free Assignment]\label{def:envyfree}
We define an envy-free assignment as a slot allocation $s(\cdot)$, where no slot is left unassigned, along with a set of prices per click $p(\cdot)$, assigning slot $s(i)$ to player $i$ and charging him $p(i)$ per click, such that for all players $i$ we have
\[ \forall i' \mbox{ with } 1\leq s(i')\leq k, u_i \geq \left\{ \begin{tabular}{ll} $\max \{ \theta_{s(i')}(v_i - p(i')),0\},$& if $\theta_{s(i')} p(i')\leq B_i$,\\0,& otherwise, \end{tabular}\right.  \]
where $u_i$ is the utility of player $i$ as defined earlier.
\end{definition}
Note that an envy-free assignment also guarantees \emph{rationality}: $p(i)\leq v_i$ for all players $i$. We say that an envy-free assignment is \emph{realizable} under a certain mechanism, if a set of bids exists such that the allocation and pricing generated by the mechanism under this set of bids matches the allocation and pricing of the envy-free assignment.

Under \BOSP, where slot allocation depends only on the value-bids and not on the budgets or budget-bids, the constraints on the bids that realize an envy-free assignment are stricter than those of a Nash equilibrium. The same holds for \BCB, as all players can pay their own bid for their slot and intuitively cannot be forced out of position by someone else' s bid\footnote{In more detail, the instability arises when some player can alter his bid to raise someone else' s price, forcing the mechanism to evict him from his slot based on budget constraints and subsequently benefiting the first player.}. Similarly, under \BCBO, a player cannot get a higher slot without paying more than the current player occupying the slot; again a realizable envy-free assignment effectively produces a Nash equilibrium. Hence, under the mechanisms \BOSP, \BCB, \BCBO, the realizable envy-free assignments form a subset of the set of Nash equilibria.  

Under \BCP\ however, the slot allocation is dependent on budgets and intuitively, one could alter the allocation to his benefit by forcing other players out of budget, hence there might exist bids that realize an envy-free assignment but do not form a Nash equilibrium.
To see this, consider the example in Figure~\ref{fig:notNE} for one slot with $\theta_1=1$ and the players budget-bidding their true budgets. While Bid (EF) is an envy-free assignment, player 1 can improve his utility by bidding according to Bid 2. Under \BCP, player 1 would still get the slot at price zero, strictly improving his utility. 

\begin{figure}[t]\center
\begin{tabular}{|c|c|c|}
\hline  & Player 1 & Player 2 \\
\hline Value & 10  & 5  \\
\hline Budget & 5 & 3 \\
\hline Bid (EF) & 5 & 4 \\
\hline Bid 2 & 3.5 & 4 \\
\hline
\end{tabular}
\caption{Example (for one slot with $\theta_1 = 1$) of an envy-free assignment which is not a Nash equilibrium under \BCP. }
\label{fig:notNE}
\end{figure}

For the other direction, it is trivial to find an example where the outcome of a Nash equilibrium is not an envy-free assignment in all mechanisms building on the intuition that someone might be envious of someone else' s higher slot but they are not able to get it at that price.

\section{The Budget-Oblivious Second-Price Auction.}\label{sec:BOSP}
\BOSP\ lacks the notions of stability defined earlier. 
\begin{theorem}\label{thm:BOSP}
There are settings where no Nash equilibrium exists under \BOSP.
\end{theorem}
\begin{proof}
We will show that \BOSP\ lacks the notions of stability defined earlier by showcasing a setting where a Nash equilibrium might not necessarily exist. Since budget-bids do not affect the allocation or pricing, they are not really relevant in this mechanism.

First, note that any set of bids where any player ends up violating his true budget constraint cannot be an equilibrium as the player can get at least zero utility by underbidding the last player. Therefore, we restrict ourselves on bids which deliver an assignment where all the budget constraints are satisfied. 

We will present a setting with 3 players and 2 slots. The high click-through rate of the first slot combined with the budget constraint will force the price of the top player at a very low point. Subsequently, the bids of the lower players need to be at such a low value that a stable outcome between them cannot be achieved.

We assume two slots with $\theta_1=1, \theta_2=0.01$. We also have 3 players with parameters as shown in Figure~\ref{fig:IPA}. Let us consider the case where $b_1\geq b_2 \geq b_3$ and the mechanism orders the players in the same order, ties accounted for, giving slot 1 to player 1 and slot 2 to player 2. It follows that $p(1)=b_2$ and $p(2)=b_3$. We will show that player 3 can improve his utility by overbidding player 2.

\begin{figure}[t]\center
\begin{tabular}{|c|c|c|c|}
\hline  & Player 1 & Player 2 & Player 3 \\
\hline Value & 10  & 9 & 14  \\
\hline Budget & 12 & 10 & 8 \\
\hline
\end{tabular}
\caption{Example of non-existence of Nash equilibrium under \BOSP for two slots with $\theta_1=1, \theta_2=0.01$.}
\label{fig:IPA}
\end{figure}
Since the budget constraints are satisfied, it follows
\begin{align*}
\theta_1 p(1) &\leq B_1 \\
b_2=p(1)&\leq B_1/\theta_1 = 12.
\end{align*}
Let us consider the case where player 3 considers overbidding player 2 but not player 1. If $b_1=b_2$ and the tie breaking rules do not make this possible consider the example where the properties of the players are exchanged accordingly. Player 3 will have utility $\theta_2 (v_3 - b_2)\geq 0.02$ strictly greater than the zero utility he currently has. Note that we are guaranteed to satisfy the budget constraint of player 3 based on the low rate of slot 2.

Now, let us consider the case where $b_3\geq b_1 \geq b_2$ and the players are ordered in this fashion. It must be the case that $b_1\leq 8$. But similarly to above, player 2 will gain positive utility by overbidding player 1 and getting slot 2.

The rest of the cases follow similarly and are omitted. 
\qed\end{proof}

Since under \BOSP, realizable envy-free assignments are a subset of Nash equilibria  it follows that:
\begin{corollary}
There are settings where no realizable envy-free assignment exists under \BOSP.
\end{corollary}

Let us point out here that in Lemma \ref{lem:existence}, stated below, we show that envy-free assignments always exist (under the assumption of distinctness  of budgets). By the above corollary, such assignments are not realizable under BOSP.

\section{The Budget-Conscious by Price Second-Price Auction.}
We now turn our attention to \BCP. We first show

\begin{theorem}
\label{thm:no-eq}
There are settings in which no Nash equilibrium exists for \BCP.
\end{theorem}

To prove Theorem \ref{thm:no-eq}, we will first show that Nash equilibria do not always exist in the special case where players are budget-bidding their true constraints and then extend the result to the general case.

\begin{lemma}
\label{lem:BASPA_budgets_public}
There are settings where players are budget-bidding their true constraints, and no Nash equilibrium exists under \BCP.
\end{lemma}
\begin{proof}
Consider two slots with $\theta_1=1,\theta_2=0.4$ and 3 players with attributes as shown in Figure~\ref{fig:EPA}. We assume that the tie-breaking ordering favors player 3 and then player 2. In order to show that a Nash equilibrium does not exist we have to consider all orderings of bids and for each such case all possible slot assignments. Intuitively, we will showcase two types of instability. If a bid is low so that the player paying it has the budget constraint satisfied then it can be easily overbid or if a bid is high then underbidding below it will force that player out of budget for the slot.

\begin{figure}[t]\center
\begin{tabular}{|c|c|c|c|}
\hline  & Player 1 & Player 2 & Player 3 \\
\hline Value & 50  & 16 & 8  \\
\hline Budget & 50 & 5 & 2 \\
\hline
\end{tabular}
\caption{Example of non-existence of Nash equilibrium under \BCP, having two slots with $\theta_1=1, \theta_2=0.4$. }\label{fig:EPA}
\end{figure}

We start by considering the case where $b_1> b_2 > b_3$ and slot 1 is assigned to player 1 and slot 2 to player 2. This means that these players can afford these slots, therefore we must have $b_2\leq 50$ and $b_3\leq 12.5$. If $b_2>5$ then player 1 can bid $b_2-\eps >b_3$, for some small $\eps$, and still get slot 1 at a lower price, since player 2 cannot afford it. If $b_2\leq 5$ then player 3 can bid $b_2+\eps<b_1$ and gain strictly positive utility. We outline in the table 

If player 1 is assigned to slot 2 because he cannot afford it and player 2 gets slot 1, we must have $b_2> 50$ and $b_3\leq 5$. If player 1 bids $b_2-\eps>5>b_3$, then player 2's bid will be the highest bid but he will not be able to afford slot 1 anymore, which will end up at player 1 for a lower price per click than before.

\begin{figure}[t]\label{fig:EPAcases}\center
\begin{tabular}{|c|c|c|c|}
\hline  Bid ordering & Slot 1 & Slot 2 & Comments \bigstrut\\
     \hline
     \multirow{4}[8]*{$b_1>b_2>b_3$} & Pl. 1  & Pl. 3 & Player 2 can bid $b_3-\eps$ and get slot 2.\bigstrut\\\cline{2-4}
           & Pl. 3  & Pl. 1 or $\emptyset$ & \begin{tabular}{c}Player 2 can bid $b_3-\eps$ \\ and get either slot 1 or slot 2.\end{tabular}\bigstrut\\\cline{2-4}
           & Pl. 2  & Pl. 3 & Player 1 can bid $b_2-\eps$ and get slot 1.\bigstrut\\\cline{2-4}
           & Pl. 3  & Pl. 2 & Player 1 can bid $b_3-\eps$ and get slot 1.\bigstrut\\
     \hline
\multirow{6}[12]*{$b_3\geq b_2 \geq b_1$} & Pl. 3  & Pl. 2 & Player 1 can bid $b_2+\eps$ and get slot 1 or slot 2.\bigstrut\\\cline{2-4}
           & Pl. 2  & Pl. 3 & Player 1 can bid $b_2+\eps$ and get slot 1 or slot 2.\bigstrut\\\cline{2-4}
           & Pl. 2  & Pl. 1 & \begin{tabular}{c} Depending on $b_2$, Player 1 can bid $b_2+\eps$ to get\\ slot 1,  or Player 3 can bid $b_2-\eps$ to get a slot.\end{tabular}\bigstrut\\\cline{2-4}
           & Pl. 1  & Pl. 2 or $\emptyset$ & \begin{tabular}{c}Player 2 can bid $b_1-\eps$ and get slot 2\\ at a lower price.\end{tabular}\bigstrut\\\cline{2-4}
           & Pl. 3  & Pl. 1 & Impossible outcome in this bid ordering.\bigstrut\\\cline{2-4}
           & Pl. 1  & Pl. 3 & Impossible outcome in this bid ordering.\bigstrut\\
     \hline
\multirow{6}[12]*{$b_1>b_3\geq b_2$} & Pl. 1  & Pl. 3 & \begin{tabular}{c} Depending on $b_3$, Player 1 can bid $b_3-\eps$ \\ to get slot 1 at a lower price, or Player 2\\ can bid  $b_3+\eps$ to get slot 2.\end{tabular}\bigstrut\\\cline{2-4}
           & Pl. 3  & Pl. 1 & Player 1 can bid $b_3-\eps$ and get slot 1.\bigstrut\\\cline{2-4}
           & Pl. 3  & Pl. 2 & Player 1 can bid $b_3-\eps$ and get slot 1.\bigstrut\\\cline{2-4}
           & Pl. 2  & Pl. 3 or $\emptyset$ & Player 1 can bid $b_3-\eps$ and get slot 1 or slot 2.\bigstrut\\\cline{2-4}
           & Pl. 1  & Pl. 2 &\begin{tabular}{c} Player 1 can bid $b_3-\eps$ and get slot 1\\ at a lower price.\end{tabular}\bigstrut\\\cline{2-4}
           & Pl. 2  & Pl. 1 & \begin{tabular}{c}Player 1 can bid $b_3-\eps$ and get slot 1\\ or slot 2 at a lower price.\end{tabular}\bigstrut\\
     \hline
\multirow{6}[12]*{$b_3\geq b_1>b_2$} & Pl. 3  & Pl. 1 & Player 2 can bid $b_1+\eps$ and get slot 1.\bigstrut\\\cline{2-4}
           & Pl. 1  & Pl. 3 & \begin{tabular}{c}Player 3 can bid $b_1-\eps$ and get slot 2\\ at a lower price.\end{tabular}\bigstrut\\\cline{2-4}
           & Pl. 1  & Pl. 2 & \begin{tabular}{c} Depending on $b_2$, Player 1 can bid $b_2-\eps$\\ to get slot 1 at a lower price, or Player 3\\ can bid $b_1-\eps$ to get slot 2.\end{tabular}\bigstrut\\\cline{2-4}
           & Pl. 2  & Pl. 1 or $\emptyset$ & Player 1 can bid $b_2-\eps$ and get slot 1.\bigstrut\\\cline{2-4}
           & Pl. 3  & Pl. 2 & Impossible outcome in this bid ordering.\bigstrut\\\cline{2-4}
           & Pl. 2  & Pl. 3 & Impossible outcome in this bid ordering.\bigstrut\\
     \hline
\multirow{7}[14]*{$b_2\geq b_1>b_3$} & Pl. 2  & Pl. 1 & \begin{tabular}{c} Depending on $b_2$, Player 1 can bid $b_2+\eps$ to get\\ slot 1, or Player 3 can bid $b_1+\eps$ to get slot 2.\end{tabular}\bigstrut\\\cline{2-4}
           & Pl. 1  & Pl. 2 &\begin{tabular}{c} Player 2 can bid $b_1-\eps$ and get slot 2\\ at a lower price.\end{tabular}\bigstrut\\\cline{2-4}
           & Pl. 2  & Pl. 3 & Impossible outcome in this bid ordering.\bigstrut\\\cline{2-4}
           & Pl. 3  & Pl. 2 & Impossible outcome in this bid ordering.\bigstrut\\\cline{2-4}
           & Pl. 1  & Pl. 3 & \begin{tabular}{c} Depending on $b_3$, Player 2 can bid $b_1-\eps$\\ to get slot 2, or Player 1 can bid \\ $b_3-\eps$ to get slot 1 at a lower price.\end{tabular}\bigstrut\\\cline{2-4}
           & Pl. 3  & Pl. 1 or $\emptyset$ &\begin{tabular}{c} Player 1 can bid $b_3-\eps$ and get slot 1\\ at a lower price.\end{tabular}\bigstrut\\
     \hline
\multirow{7}[14]*{$b_2>b_3\geq b_1$} & Pl. 2  & Pl. 3 &  Player 1 can bid $b_3+\eps$ and get a slot.\bigstrut\\\cline{2-4}
           & Pl. 3  & Pl. 2 & Player 1 can bid $b_3+\eps$ and get slot 1.\bigstrut\\\cline{2-4}
           & Pl. 2  & Pl. 1 & Impossible outcome in this bid ordering.\bigstrut\\\cline{2-4}
           & Pl. 1  & Pl. 2 or $\emptyset$ & \begin{tabular}{c}Player 2 can bid $b_3-\eps$ and get slot 1\\ or slot 2 at a lower price.\end{tabular}\bigstrut\\\cline{2-4}
           & Pl. 3  & Pl. 1 & Player 2 can bid $b_3-\eps$ and get slot 1.\bigstrut\\\cline{2-4}
           & Pl. 1  & Pl. 3 & Player 2 can bid $b_3-\eps$ and get slot 1.\bigstrut\\
     \hline
\end{tabular}
\caption{Possible cases that can occur in the example  in Lemma~1.}\label{fig:EPAcases}
\end{figure}

The rest of the cases follow similarly. We note that the first slot cannot be left unassigned as the last player in the bid ordering can always afford it at zero price. The only way for the second slot to be left unassigned is if both the first and second player in the respective bid ordering cannot afford either slot. In Figure~\ref{fig:EPAcases}, we outline  all possible cases that can occur in the example described in the lemma. In each case, we briefly outline why the profile cannot be in equilibrium.

\qed\end{proof}

\begin{lemma}\label{lem:BASPA_budget_bids}
Under \BCP, if a Nash equilibrium exists then a Nash equilibrium also exists where the players are budget-bidding their true constraints.
\end{lemma}
\begin{proof}
Suppose a Nash equilibrium set of bids with $\langle b_i,g_i\rangle$ the bid and budget-bid of player $i$ and $s(i)$ is the slot allocated to him. Let us consider the case where player $i$ changes his budget-bid to $B_i$. We argue that the same slot allocation can be achieved. First observe that any change must include a change in the slot or price of player $i$.

If $g_i>B_i$ and it results in a slot change it must be the case that he got a lower slot. Furthermore, $s(i)$ was a slot that he could not afford therefore we could not be in a Nash equilibrium. The same holds if player $i$ simply ends up paying less for the same slot $s(i)$.

If $g(i)<B(i)$ and we have a change to slot $s'(i)$, it must be that $s'(i)$ is a higher slot than $s(i)$. But the same allocation can be maintained if $i$ bids accordingly lower to get $s(i)$ at the same price as before. If this is not possible due to ties, it means $s(i)$ is at the same price as other slots and $i$ was ``losing'' some slot(s) due to budget. But since he can afford one of these higher slots based on his real budget $B_i$, he has strictly greater utility in a higher slot and we wouldn't be in an equilibrium. If no change in the slot of player $i$ occurs but only on his price, it is easy to see that the same argument holds and he can get the same price as before.

Observe that we can also adjust the bids of all other players in a ``shift-up'' fashion to achieve the same prices for everyone as before. We are guaranteed to achieve the same allocation since everyone will be asked to pay the same price for the same slot as before.

We continue to adjust all budget-bids one by one until all players are budget-bidding their true budget and we will still be in a Nash equilibrium. But the bids at that stage would yield a Nash equilibrium.
\qed\end{proof}

Theorem \ref{thm:no-eq} follows
from Lemmas~\ref{lem:BASPA_budgets_public} and~\ref{lem:BASPA_budget_bids}.

\clearpage

Despite the non-existence of Nash equilibria, the {\BCP} mechanism does possess other stability properties, in particular, all envy-free assignments are realizable under the mechanism (and there exists at least one such assignment).
\begin{theorem}
\label{basic}
There exists an envy-free assignment which is realizable under \BCP.
\end{theorem}

The proof of the Theorem is based on Lemmas \ref{lem:existence} and \ref{lem:efBASPA}.

\begin{lemma}
\label{lem:existence}
Under budget constraints, there always exists an envy-free assignment.
\end{lemma}
\begin{proof}
We describe a process that iteratively lowers the prices of slots from infinity and eventually assigns all slots while satisfying the envy-free conditions. We associate a value $U_i$ for each player $i$, intuitively denoting (but not always equal to) their current utility. Originally all $n$ players are unassigned with $U_i$, for all $i$, set to 0. Slots not assigned to a player will be referred to as {\em free}.

During the process, players may move from slot to slot and prices may decrease, but never strictly increase. When a player $i$ is assigned to a slot $s$ with price $p_s$, we update $U_i$ by $\theta_s (v_i - p_s)$ and we will guarantee the assignments respect the budget constraints. It will be evident from the process that a particular $U_i$ may only increase. We will prove that at the end of the process all slots will have been assigned to a player while the envy-free conditions as defined in (\ref{def:envyfree}) will be satisfied and so the proof will be complete.

We start the process by assigning to every slot the same positive price $\infty$ such that no player can afford any slot at that price.

Given the current price $p_s$ assigned to a free slot $s$, we say that a player $i$ (assigned to some slot or not) {\em barely wants} $s$ if (a) $i$ can afford $s$ at its current price $p_s$, $\theta_s p_s \leq B_i$ and (b) $i$'s current $U_i$ is equal to his utility if he were to be assigned to $s$. In case (a) $i$ can afford $s$ at its current price and (b) $i$'s  current $U_i$ is strictly less than his utility if he were to be assigned to $s$, we say that $i$ {\em envies} slot $s$. Finally, we say that $i$ wants $s$ if he barely wants it or he envies it.

Notice that if during a price decrease of a slot, a point is reached where a player wants the slot, whereas he does not want it at a higher price, this may be either because the player barely wants that slot at the current price, or alternatively, the current price is such that the player can just afford this slot.

Now consider an arbitrary free slot $s$ and decrease its current price $p_s$ until a point $p$ is reached where some player (assigned or not)  wants $s$ at  price $p$ (notice that no strict decrease  of $p_s$ might be necessary to attain such desire, as some player may want $s$ at the price $p_s$). Then we distinguish three cases which are examined in order:
\begin{enumerate}
\item There is a player $i$ who envies $s$ at price $p$. Assign  
$s$ to $i$ and set the former's price to $p$, freeing up his previous assigned slot if he had any. Notice that thus if during the lowering of the price of a free slot $s$ a player $i$ wants $s$  and did not previously want it, it is because $i$ can just afford $s$. Therefore by the distinctness of budgets, there is a unique such player, so $i$'s choice is uniquely determined. 
  
\item There are no players who envy $s$ at price $p$ but there are unassigned players that barely want it. These are players who can afford $s$ at the current price $p$ and either $p$ is equal to their value $v_i$ or their utility at $s$ for $p$ is equal to the non-zero value of $U_i$, a situation which can occur from the procedure described in later steps. We assign $s$ to an arbitrary player of this type.

\item There are no unassigned  players who want $s$, no assigned player who envies $s$ but there is a set $I$ of assigned players who want $s$. Obviously, all members of $I$ barely want  $s$.

Notice that on one hand we cannot further lower the price of $s$, lest  envying of $s$ by the players of $I$ concurrently occurs, an undesirable situation especially if $I$ contains more than one element. On the other hand, assigning $s$ to a member of $I$  without further lowering the price $p$  may result in a round robin moving around of assigned players with no price lowering. To avoid this situation, we work as described below: Let $S$ be the set of slots occupied by players in $I$, together with the currently empty slot $s$. We {\em concurrently} decrease the price of {\em all} slots in $S$  {\em proportionally}, i.e. we decrease the price per click at each such slot $t \in S$ by $X/\theta_{t}$ for some $X$. We start from $X=0$ and increase it, until one or more players  want one of the slots. Notice that during this process, all players of $I$ can obtain the same utility from at least two slots in $S$. We distinguish different scenarios.
\begin{enumerate}
\item Only assigned players barely want slots in $S$. These players join $I$, their slots are added to $S$, and we resume the concurrent lowering with the new sets.
\item At most one assigned player envies a slot or at most one unassigned player wants a slot, but not both. In this case, we can assign this slot to the player in question and shift the players of $I$ to slots where they obtain the same utility. This can be done since every player has another slot they can receive the same utility and these paired slots must form a ``chain'' to the empty slot $s$.
\item There is more than one player in the set of unassigned players wanting a slot together with assigned players envying a slot. In this case, we still assign the slots as requested by these players with priority to the envious players. We still potentially shift the players of $I$ to a different slot so that $s$ gets assigned as well but we will not be able to accommodate all of the players who lost their slots; these players will become unassigned.
\end{enumerate}
\end{enumerate}

We continue this process until no slot remains unassigned since there are more players than slots. \if{On the steps of the process that do not end with one more slot assigned, the newly free slots are not envied by anyone, therefore we are guaranteed a strict price decrease when they are examined by the process later on and a player envying them. It is easy to also see that the points where envying occurs are finite, therefore we are guaranteed that this process will terminate in finite time.}\fi
At every step of the process we respect the budget constraints when players are assigned to slots. Also, the envy-free constraints are satisfied except for the players that become unassigned in the last step of the process above, i.e. the unassigned players with positive $U_i$. We will show that when the process terminates no such player will remain unassigned.

Let us focus on one of the players, call him $A$, that will become unassigned in the last step above. This player could receive the same utility from at least two slots, say $1$ and $2$ of $S$ with prices $p(1)$ and $p(2)$ respectively, such that $p(1)\geq p(2)$. If $p(1)=p(2)$, it must be that $p(1)=p(2)=v_A$ and $A$ was getting zero utility, therefore we can safely leave him unassigned and from now focus on the case where $p(1)> p(2)$. 
First, let us note that no pair of evicted players can exist so that player $X$ had equal utility in slots $s_1$, $s_2$ and player $Y$ had equal utility in slots $s_1$, $s_3$ with $s_1$ being the most expensive slot among the three; the player in the cheapest slot would be envious of the other one. Therefore, it is possible to find a unique slot for each evicted player such that he received the same utility from this slot and another cheaper one. These will be the referred slots 1 and 2 for each evicted player.
We will show that the player $B$ ending up at slot $1$ could not be originally unassigned with $U_B< \theta_2(v_B- p(2))$. 

Since $\theta_1 (v_A -p(1))=\theta_2(v_A-p(2))$ and $p(1)>p(2)$ we have $\theta_1 > \theta_2$. Therefore, $\theta_1 p(1) > \theta_2 p(2)$ and if $B$ can afford slot 1, he can also afford slot 2, even at a price lower than $p(2)$. Also $v_B \geq p(1) > p(2)$ which means that if $B$ was unassigned with $U_B < \theta_2(v_B-p(2))$ he would have envied slot 2 at an earlier stage, contradicting our lowering process. We conclude that $B$ was either assigned to a slot or unassigned with $U_B\geq \theta_2(v_B-p(2))$. More interestingly, there cannot exist anymore an unassigned player that envies slot 2 at price $p(2)$. We will now show that this guarantees that $A$ will eventually get a slot with utility at least $\theta_2 (v_A - p(2))$, as much as he was getting before he was ``evicted''.

If $B$ was assigned in the previous scenario, let us assume it was at slot 3 at price $p(3)$. If $B$ is a member of $I$ or even originally in slot 1, we consider slot 3 to be the other slot at which $B$ could receive the same utility. We know that $U_A=\theta_1 (v_A - p(1))=\theta_2 (v_A - p(2))$ and $\theta_1 (v_B - p(1))\geq \theta_2 (v_B -p(2))$ which, with $\theta_1>\theta_2$ and $p(1)>p(2)$, leads us to the conclusion that $v_B> v_A$. We know $A$ did not envy 3 at $p(3)$ and $B$ does not envy 2 at $p(2)$, $\theta_2 (v_A - p(2))\geq \theta_3 (v_A - p(3))$ and  $\theta_2 (v_B - p(2))\leq \theta_3 (v_B - p(3))$, which gives us $\theta_3 > \theta_2$ and $p(3)>p(2)$. It is now easy to see that $A$ could get utility $U_A$ from slot 3 even at a price greater than $p(2)$. For example, if we start lowering the price of 3, $A$ is going to request it before the price reaches $p(2)$ unless someone else takes it beforehand. But if someone does, say an assigned player 4, it must be the case that he did not envy slot 2 at $p(2)$. Repeating the same arguments as above, we conclude that $\theta_4>\theta_2$ and $p(4)>p(2)$, which means that $A$ would be also happy at slot 4. The same applies if no price lowering of slot 3 occurs and if the player in 4 is actually a member of $I$ who got slot 3 at the same moment; inductively and since we were forced to ``evict'' a player of $I$, we are bound to reach a case where one of the players getting slots desired by $A$ was not in $I$ or was unassigned.

Let us now assume that $B$ was unassigned. As before, we get $v_B>v_A$ based on the fact that $B$ weakly prefers 1 to 2. Also, since $B$ didn't want 2 before, $U_B\geq\theta_2 (v_B-p(2))>\theta_2 (v_A -p(2))$. But that means that there is another free slot available, intuitively the one ``reserved'' for $B$, where player $A$ could get utility at least $\theta_2 (v_A - p(2))$. Furthermore, he will be able to get it at an affordable price since $B$ originally had it. Combining all of the arguments above inductively we conclude that no player with positive $U_i$ can be left unassigned at the end of the process.

All the steps of the process end with at least one of the following outcomes. Set $I$ is strictly increased, the number of free slots strictly decreases or an envying player receives a slot. Since, players' $U$ can only increase through the process, the value points where players can be become envious of slots are finite and we are guaranteed termination of the process.
\qed\end{proof}

 Before proving Lemma~\ref{lem:efBASPA}, we need a technical lemma. 
 \begin{lemma}\label{lem:efconditions}
Consider an envy-free assignment with slot allocation $s(\cdot)$ and prices $p(\cdot)$. For any two players $i, j$ that are awarded a slot by the assignment, if $s(i)<s(j)$ and $p(j)>p(i)$ then $p(j)\theta_{s(i)}>B_j$, i.e., player $j$ cannot afford the slot of player $i$ with his current price.
\end{lemma}
\begin{proof}
Suppose towards a contradiction, that for some $i,j$ with $1\leq s(i),s(j)\leq k$,  $s(i)<s(j)$ and $p(j)>p(i)$ with $p(j) \theta_{s(i)} \leq B_j$. This also means
\begin{equation} \label{budget2} p(i) \theta_{s(i)} \leq B_j.\end{equation}
First observe that since $p(j) >0$, $\theta_{s(i)} > \theta_{s(j)}>0$ and $p(i) < p(j)$, we have \begin{equation} \label{budget3} p(i)\theta_{s(i)} -p(j)\theta_{s(j)} < p(j)(\theta_{s(i)} -\theta_{s(j)}).\end{equation}
Also recall that rationality follows from the envy-free definition, thus $p(j) \leq v_j$. Therefore by equation~\eqref{budget3} above we get $p(i)\theta_{s(i)} -p(j)\theta_{s(j)} < v_j(\theta_{s(i)} -\theta_{s(j)})$ and therefore
$\theta_{s(i)}(v_j-p(i)) > \theta_{s(j)}(v_j-p(j))$, which in conjunction with equation \eqref{budget2} contradicts the envy-free constraints.
\qed\end{proof}

\begin{lemma}\label{lem:efBASPA}
Any envy-free assignment is realizable under \BCP.
\end{lemma}
\begin{proof}
Recall that we have assumed that there is an a priori strict ordering of the players used to break ties between equal bids in any auction.

We define budget-bids to be equal to the real budgets of the players. It suffices to define value-bids $b_i, i=1, \ldots n$ so that \BCP\ produces the same allocation of slot and prices to the players as the given envy-free assignment. For simplicity, we assume players not assigned to slots are assigned to virtual slots with zero \CTR, following the order of their bids.

Given $s(\cdot)$ and $p(\cdot)$ from the envy-free assignment, let
$$\pi: \{1,\ldots, n\}\mapsto \{1,\ldots, n\}$$ be a rearrangement of the players so that $p({\pi(1)}) \geq \cdots \geq p({\pi(n)}).$ In case of ties, $\pi$ respects the a priori ordering of the players. Let now $b_{\pi(1)}$ be a number strictly greater than $p({\pi(1)})$ and also $b_{\pi(i)}$ be $p({\pi(i-1)})$ for $i=2, \ldots, n$.
We claim now that with this bid vector, \BCP\ produces the same assignment as the envy-free assignment. Obviously \BCP\ examines the players in the order $\pi(1), \ldots, \pi(n)$ and assigns the prices
$p({\pi(1)}) \geq \cdots \geq p({\pi(n)}$)  to them.

Suppose towards a contradiction that the auction assigns a player $i$ to a slot $s(j) \neq s(i)$. Assume also that $i$ is the first (in the order of descending bids) player ``misplaced" by the auction. It is not possible to have that $s(i)<s(j)$ because by the auction's budget constraint, player $i$ affords slot $s(i)$ and so he would have been assigned by the auction there or higher, rather than at slot $s(j)$. Assume now that $s(j) <s(i)$. Since player $i$ is assigned to slot  $s(j)$ by \BCP, player $i$ affords slot  $s(j)$. Therefore, by Lemma~\ref{lem:efconditions},  $p(i) \leq p(j)$. We distinguish two cases:
\begin{itemize}
\item If $p(i) < p(j)$, then the rearrangement $\pi$ would place $i$ later  than $j$. Since obviously player $j$ is also misplaced by the auction (his slot is taken by $i$), we contradict the priority in the choice of $i$.

\item    If $p(i) = p(j)$, then player $i$ envies the higher slot $s(j)$, which he can afford since the auction placed him there, again a contradiction of the envy-free assignment.
\end{itemize}
\qed\end{proof}
This completes  the proof of Theorem \ref{basic}.
\section{Budget-Conscious by Bid and Best Offer Second-Price Auctions.}

Recall that under \BCB\ and \BCBO, realizable envy-free assignments are a subset of Nash equilibria. We are going to show that for these two mechanisms, Nash equilibria exist, by establishing that envy-free assignments are realizable in both \BCB\ and \BCBO. Let us first focus on \BCB. 
\begin{theorem}\label{thm:BCSPA}
Under \BCB, there is always a Nash equilibrium that produces an envy-free assignment.
\end{theorem}
\begin{proof}
First, an envy-free assignment always exists regardless of the mechanism, as stated earlier. Given such an assignment (allocation $s(\cdot)$ and pricing $p(\cdot)$) we can construct value-bids and budget-bids that produce the same slot allocation and pricing under \BCB\ and these will also form a Nash equilibrium. Similarly to the proof of Lemma~\ref{lem:efBASPA}, let us define a permutation of the players $\pi(\cdot)$ such that $p({\pi(1)}) \geq \cdots \geq p({\pi(n)})$ where the players not awarded slots are ordered in an arbitrary fashion.

For $i\in[2,\ldots,k+1]$, we set

\begin{eqnarray*}
b_{\pi(i)}&=& p(\pi(i-1))\\
g_{\pi(i)}&=& \theta_{s(\pi(i))} p(\pi(i-1))
\end{eqnarray*}
while $b_{\pi(1)}$ can be set appropriately high with $g_{\pi(1)}=\theta_1 b_{\pi(1)}$. The losing players' bids are set lower than the value-bid and budget-bid of player $\pi(k+1)$.

It is clear that the prices assigned by the mechanism are going to match those of the envy-free assignment. Let us now consider how the slots will be allocated. By the budget-bids definition, observe that any of the winning players can afford, according to his budget-bid, the slot he was awarded in the envy-free assignment but not any of the higher slots. But this means that when the mechanism considers him, he will be allocated that particular slot since by a simple inductive argument all previous players are assigned at their respective slots and the ``right'' slot will be available. We also take advantage of our assumption that the click-through rates are distinct.

We conclude that the players are going to be allocated and charged as in the envy-free assignment. As previously noted, e.g. recall the discussion in Subsection \ref{subsec:stable}, bids that produce an envy-free assignment form a Nash equilibrium as well for this mechanism.
\qed\end{proof}
\begin{theorem}\label{thm:BLSPA}
Under \BCBO, there is always a Nash equilibrium that produces an envy-free assignment.
\end{theorem}
\begin{proof}
Similarly to the previous theorem, we consider an envy-free assignment with allocation $s(\cdot)$ and pricing $p(\cdot)$. 

This time we define a permutation of the players $\pi(\cdot)$ according to the slot ordering in the envy-free assignment with the losing players ordered in an arbitrary fashion. Formally, $s(\pi(i))=i$ for $i=1\ldots k$.

For $i\in[2,\ldots,k+1]$, we set

\begin{eqnarray*}
b_{\pi(i)}&=&  p(\pi(i-1))  \theta_{s(\pi(i-1))} / \theta_{s(\pi(i))} = p(\pi(i-1)) \theta_{i-1}  / \theta_{i}\\
g_{\pi(i)}&=& \theta_{s(\pi(i-1))} b_{\pi(i)} = \theta_{i-1} p(\pi(i-1))
\end{eqnarray*}
and $b_{\pi(1)}$ is set to an appropriately high value, $g_{\pi(1)}=\theta_1 b_{\pi(1)}$, while the losing players' bids are set lower than the respective bids of player $\pi(k+1)$.

We will derive an inductive argument by assuming all slots from the first down to $i-1$ have been assigned and charged for as in the envy-free assignment. We will argue that player $\pi(i)$ will win slot $i$. Indeed, for purposes of contradiction, let us consider some other player $\pi(j), j>i$ that wins slot $i$. It is easy to see that the mechanism will consider $p(\pi(j-1)) \theta_{j-1}  / \theta_i$ as the amount per click offered by player $j$ while the amount offered by player $i$ is $p(\pi(i-1)) \theta_{i-1} /\theta_i$. We have 
\begin{eqnarray*}
\theta_{j-1} p(\pi(j-1)) \geq \theta_{i-1} p(\pi(i-1)).
\end{eqnarray*}
But that means that $p(\pi(j-1)) \theta_{j-1}/\theta_{i-1}$, the price $\pi(j-1)$ would offer for slot $i-1$, is higher than $p(\pi(i-1))$ a contradiction to our inductive assumption that slots up to $i-1$ have been charged as in the envy-free assignment.

It remains to show that player $\pi(i)$ will be charged $p(\pi(i))$. Similarly to above, player $i+1$ offers per click $p(\pi(i))$ while some other player $\pi(j),j>i+1$, would offer $p(\pi(j-1)) \theta_{j-1}  / \theta_i$. If this amount is higher, we would have

\begin{eqnarray*}
\theta_{j-1} p(\pi(j-1)) \geq \theta_{i} p(\pi(i)).
\end{eqnarray*}
But this would mean that player $j-1$ can afford, using his real budget, slot $i$ and he would be envious of that slot since he is paying more at slot $j-1$ contradicting our envy-free assignment constraints.

We conclude that using the above mentioned bids, an envy-free assignment is realizable under \BCBO\ and consequently these bids form a Nash equilibrium.
\qed\end{proof}

Note that in both theorems we have set the budget-bids of the players different from their true budgets. It turns out this is necessary, a perhaps surprising result as having budgets publicly known was necessary for truthful mechanisms in related work~\cite{Dobzinski08,ABHLT13}. 
\begin{theorem}
There are settings with public budgets, where a Nash equilibrium does not exist under both \BCB\ and \BCBO.
\end{theorem}
\begin{proof}
We work below under \BCB. The same counterexample, with the same proof, works for \BCBO\ as well (always under the assumption that the true budgets are declared).

We assume that we have four players and three slots such that $\theta_1> \theta_2 > \theta_3$.

Let the (distinct and true) budgets of the four players be $B_1 > B_2 > B_3 > B_4$ (i.e. we assume that the players are ordered according to decreasing budget and also that they are named  after their rank in this  ordering ---no assumption as to which player gets which slot).

Let  $b_1, b_2, b_3$ and $b_4$ be the bids of the four players (no assumption with respect to the ordering of the bids).

Let $p_i^j$ be the highest per click price at which player $i$ can afford  slot $j$, i.e. $ p_i^j = B_i/\theta_j$.

Assume that the twelve values $p_i^j$ are pairwise distinct (this can be fixed by the choice of the parameters of the example). Obviously then $p_i^j < p_{i'}^j$ if $i > i'$ and $p_i^j < p_i^{j'}$ if $j < j'$.

Assume that the private valuations of all players are very high, higher than $p_1^3$,  and such that any player has a strictly positive incentive to move to a higher slot than where he presently is, given that (a) he can afford the higher slot and (b) he ends up paying a  price  $\leq p_1^3$ for it.

Assume that ties are  resolved in the order of the (public) budgets (players with higher budgets get higher priority in a tie).

Assume that $p_4^1 < p_3^1 < p_2 ^1 < p_1^1 <  p_4^2 < p_3^2 < p_2^2 < p_1^2 < p_4^3   < p_3^3     <  p_2^3 < p_1^3$, easily achievable by setting the budgets close enough.

We are going to give a series of claims regarding the purported NE. The claims will lead to a contradiction if we assume that a NE exists.

\begin{claim}
In a NE no slot is empty.
\end{claim}
\begin{proof} Indeed, if  slot $j$ is empty,  and player $i$  has not been assigned a slot, then player $i$ can unilaterally change his bid to $p_i^j$. Thus he will get slot $j$, strictly increasing his utility. \end{proof}
\begin{claim}
In a NE, player $i$ gets slot $i$ for $i= 1, 2, 3$.
\end{claim}
\begin{proof} Let us start by proving that player 1 gets slot 1. Otherwise, player 1 can alter his bid to be  $p_1^1$ and thus will get slot 1, since the other players either have a lower bid than $p_1^1$, or cannot afford slot 1.  Now to show that player 2 gets slot 2, observe that otherwise player 2 can change his bid to be $p_2^2$, and thus he will either (a)  get slot 1, and since he did not have slot 1 before the change,  he strictly increases his utility, or otherwise  (b)  get slot 2, because players 2,3 and 4 will either have a lower bid than $p_2^2$, or cannot afford slot 2, and thus again strictly increases his utility. Similarly, we show that  player 3 gets slot 3. 
By the previous claim we immediately get that in a NE $b_i \leq p_i^i$ for $i = 1,2,3$.\end{proof}
\begin{claim}
In a NE, for $i= 1,2,3$ $b_i \geq p_{i+1}^{i}$, and therefore by the previous claim, $b_i \in [p_{i+1}^{i}, p_{i}^i]$.
\end{claim}
\begin{proof} Suppose that $b_i < p_{i+1}^{i}$. Then player $i+1$ can change his bid to $ p_{i+1}^{i}$ and get slot $i$. \end{proof}

By the previous claims, we have that in a NE  $p_2^1 \leq b_1  \leq p_1^1 <  p_3^2 \leq b_2 \leq p_2^2  < p_4^3 \leq  b_3 \leq p_3^3$. Now we examine two cases:
\begin{itemize}
\item $b_4 < b_2$. Then if player $3$  lowers his bid to $b_2$, he will still get slot $3$  and at a strictly lower price than before, namely $<b_2$ rather than $b_2$  (take into account the rule for ties).
    \item $b_4 \geq b_2$. Since $b_4>p_2^4$, player 4 cannot get slot 2. Then if player $2$  lowers his bid to $b_1$, he will still get slot $2$  and at a strictly lower price than before, namely 0 rather than $b_1$ (take into account the rule for ties).
\end{itemize}
In both cases, a contradiction to the assumption of existence of a Nash equilibrium is reached.
\qed\end{proof}

As  envy-free assignments realizable under  \BCB\ or \BCBO\  form a subset of Nash equilibira of these two mechanisms, respectively,  the Theorem above implies that there are settings where  neither of these  mechanisms can realize any envy-free assignment (guaranteed to exist by Lemma \ref{lem:existence}), unless players are allowed to bid non-true budgets. 
\section{Discussion}
We have studied the problem of introducing budget constraints in GSP auctions. Since this cannot be done in a single straightforward way, we have proposed and explored different natural approaches. We have investigated the existence of Nash equilibria and envy-free assignments, and our results demonstrate that small changes in the way budget constraints are handled may affect the stability of these mechanisms significantly.

The consideration of several variants of  mechanisms, introduced not out of idle curiosity, but as  representations of  all the natural answers to natural questions raised by the introduction of budgets, and the examination of their properties and differences is what we consider as our  primary contribution in this work. 
We believe our work can serve as a starting point for studying further the properties of GSP auctions under budget constraints. Although we studied these auctions in the context of sponsored search, second-price auctions are widely used in many different settings both in off-line and on-line scenarios. As such, our results are applicable in a much wider context. In the area of sponsored search, a further step is required towards the more accurate modeling of the deployed systems. In practice, a player can transition between slots during a time period, as players are moderated according to their budget depletion rate. 
Similarly to the majority of related work on keyword auctions with budgets, we chose to study the static setting first both as a stepping stone and in its own interest for settings outside of keyword auctions. An analysis on the effects of budget constraints in a dynamic setting that would extend upon our results, would contribute towards a more accurate modeling of sponsored search auctions.
Another interesting direction for future research is to evaluate the performance of these mechanisms in terms of the generated welfare.
This type of Price of Anarchy analysis for non-truthful auctions was initiated in \cite{CKS08}, and for certain settings with budget constraints, some results have been recently obtained in \cite{ST13}.

\section{Acknowledgements}
We would like to thank Konstantinos Gavriil for pointing out to us the counterexample of Figure \ref{fig:notEF}, that without distinct budgets, envy-free assignments may fail to exist. We also want to thank Giorgos Birbas for valuable discussions during the preparation of this work.

\end{document}